\documentclass[pra,twocolumn,showpacs,floatfix]{revtex4}
\usepackage[pdftex]{graphicx}
\usepackage{amsmath,amsfonts,amsbsy,amssymb,amsthm}

\usepackage{mathrsfs,bbm}

\newcommand{\ket}[1]{\lvert #1 \rangle}
\newcommand{\bra}[1]{\langle #1 \rvert}
\newcommand{\ketbra}[2]{\ket{#1}\bra{#2}}
\newcommand{\braket}[2]{\langle #1 \rvert #2 \rangle}
\newcommand{\Tr}[1]{\operatorname{Tr}\bigl[#1\bigr]}

\newcommand{\abs}[1]{\lvert #1 \rvert}

\newcommand{\expect}[1]{\langle #1 \rangle}

\newtheorem{thm}{Theorem}
\newtheorem{cor}{Corollary}
\numberwithin{cor}{thm}
\begin{document}
\title{Single shot parameter estimation via continuous quantum measurement}
\author{Bradley A. Chase}
\email{bchase@unm.edu}
\author{JM Geremia}
\email{jgeremia@unm.edu}
\affiliation{Quantum Measurement \& Control Group, Department of Physics \& Astronomy, The University of New Mexico, Albuquerque, New Mexico 87131 USA}
\date{\today}
\begin{abstract}
	We present filtering equations for single shot parameter estimation using continuous quantum measurement.  By embedding parameter estimation in the standard quantum filtering formalism, we derive the optimal Bayesian filter for cases when the parameter takes on a finite range of values.  Leveraging recent convergence results [van Handel, arXiv:0709.2216 (2008)], we give a condition which determines the asymptotic convergence of the estimator.  For cases when the parameter is continuous valued, we develop \emph{quantum particle filters} as a practical computational method for quantum parameter estimation.  
\end{abstract}
\maketitle
\section{Introduction}
Determining unknown values of parameters from noisy measurements is a ubiquitous problem in physics and engineering.    In quantum mechanics, the single-parameter problem is posed as determining a coupling parameter $\xi$ that controls the evolution of a probe quantum system via a Hamiltonian of the form $H_{\xi} = \xi H_0$ \cite{Helstrom:1976a,Holevo:1982a,Braunstein:1994a,Braunstein:1995a,Giovannetti:2004a,Giovannetti:2006a,Boixo:2007a}.  Traditionally, an estimation procedure proceeds by (i) preparing an ensemble of probe systems, either independently or jointly; (ii) evolving the ensemble under $H_{\xi}$; (iii) measuring an appropriate observable in order to infer $\xi$.  The quantum Cram\`{e}r-Rao bound \cite{Cramer:1946a,Helstrom:1976a,Holevo:1982a,Braunstein:1994a,Braunstein:1995a} gives the optimal sensitivity for \emph{any} possible estimator and much research has focused on achieving this bound in practice, using entangled probe states and nonlinear probe Hamiltonians \cite{Nagata:2007a,Pezze:2007a,Woolley:2008a}. 

Yet, it is often technically difficult to prepare the exotic states and Hamiltonians needed for improved sensitivity.    Instead, an experiment is usually repeated many times to build up sufficient statistics for the estimator.  In contrast, the burgeoning field of continuous quantum measurement \cite{Bouten:2006a} provides an opportunity for on-line \emph{single-shot} parameter estimation, in which an estimate is provided in near real-time using a measurement trajectory from a single probe system.  Parameter estimation via continuous measurement has been previously studied in the context of force estimation \cite{Verstraete:2001a} and magnetometry \cite{Geremia:2003a}.  Although Verstraete et. al develop a general framework for quantum parameter estimation, both of \cite{Verstraete:2001a,Geremia:2003a} focus on the readily tractable case when the dynamical equations are linear and the quantum states have Gaussian statistics.  In this case, the optimal estimator is the quantum analog of the classical Kalman filter \cite{Belavkin:1999a,Kalman:1960a,Kalman:1961a}.

In this paper, we develop on-line estimators for continuous measurement when the dynamics and states are not restricted.  Rather than focusing on fundamental quantum limits, we instead consider the more basic problem of developing an actual parameter filter for use with continuous quantum measurements.  By embedding parameter estimation in the standard quantum filtering formalism \cite{Bouten:2006a}, we construct the optimal Bayesian estimator for parameters drawn from a finite dimensional set.  The resulting filter is a generalized form of one derived by Jacobs for binary state discrimination \cite{Jacobs:2006a}.  Using recent stability results of van Handel \cite{vanHandel:2008a}, we give a simple check for whether the estimator can successfully track to the true parameter value in an asymptotic time limit.  For cases when the parameter is continuous valued, we develop \emph{quantum particle filters} as a practical computational method for quantum parameter estimation.  These are analogous to, and inspired by, particle filtering methods that have had much success in classical filtering theory \cite{Doucet:2001,Arulampalam:2002a}.  Although the quantum particle filter is necessarily sub-optimal, we present numerical simulations which suggest they perform well in practice.  Throughout, we demonstrate our techniques using a single qubit magnetometer.  

The remainder of the paper is organized as follows.  Section \ref{sec:quantum_filtering} reviews quantum filtering theory.  Section \ref{sec:finite} develops the estimator and stability results for a parameter from a finite-dimensional set.  Section \ref{sec:infinite} presents the quantum particle filtering algorithm, which is appropriate for estimation of a continuous valued parameters. Section \ref{sec:conclude} concludes.

\section{Quantum filtering}
\label{sec:quantum_filtering}
In this section, we review the notation and features of quantum filtering and quantum stochastic calculus, predominantly summarizing the presentation in \cite{Bouten:2006a}, which provides a more complete introduction.  In the general quantum filtering problem, we consider a continuous-stream of probe quantum systems interacting with a target quantum system.  The probes are subsequently measured and provide a continuous stream of measurement outcomes.  The task of quantum filtering is to provide an estimate of the state of the target system given these indirect measurements.  In the quantum optics setting, the target system is usually a collection of atomic systems, with Hilbert space $\mathcal{H}_A$ and associated space of operators $\mathcal{A}$.  The probe is taken to be a single mode of the quantum electromagnetic field, from which vacuum fluctuations give rise to white noise statistics.

In the limit of weak atom-field coupling, the joint atom-field evolution is described by the following quantum stochastic differential equation (QSDE)
\begin{equation} \label{eq:stochastic_propogator}
	dU_t = \left( L dA_t^{\dag} - L^{\dag} dA_t  - \frac{1}{2}L^{\dag}L dt - iHdt\right)U_t ,
\end{equation}
where $L \in \mathcal{A}$ is an atomic operator that describes the atom-field interaction and $H \in \mathcal{A}$ is the atomic Hamiltonian.  The interaction-picture field operators $dA_t,dA_t^{\dag}$ are quantum white noise processes with a single non-zero It\^{o} product $dA_tdA_t^{\dag} = dt$.

For any atomic observable $X_A \in \mathcal{A}$, the Heisenberg evolution or quantum flow is defined as $(X_A)_t = j_t(X_A) = U_t^{\dag}(X_A\otimes I ) U_t$.  Application of the It\^{o} rules gives the time evolution as
\begin{equation} \label{eq:quantum_flow}
	dj_t(X_A) = j_t(\mathcal{L}[X_A])dt + j_t([L^{\dag},X_A])dA_t + j_t([X_A,L])dA_t^{\dag}
\end{equation}
with Lindblad generator
\begin{equation} \label{eq:lindblad_generator}
	\mathcal{L}[X_A] = i[H,X_A] + L^{\dag}X_AL -\frac{1}{2}L^{\dag}LX_A - \frac{1}{2}X_AL^{\dag}L .
\end{equation}
Similarly, the observation process, which we take to be homodyne detection of the scattered field, is given by $M_t = U_t^{\dag}(A_t + A_t^{\dag})U_t$.  The It\^{o} rules give the corresponding time evolution
\begin{equation} \label{eq:measurements}
	dM_t = j_t(L + L^{\dag})dt + dA_t + dA_t^{\dag} .
\end{equation}  
Together, \eqref{eq:quantum_flow} and \eqref{eq:measurements} are the system-observation pair which define the filtering problem.  The quantum flow describes our knowledge of how atomic observables evolve exactly under the joint propagator in \eqref{eq:stochastic_propogator}, but it is inaccessible since the system is not directly observed.  Nonetheless, the scattered fields as measured in \eqref{eq:measurements} carry information about the atomic system, providing a continuous measurement of the observable $L + L^{\dag}$, albeit corrupted by quantum noise.  The quantum filtering problem is to find $\pi_t[X_A] = \mathbbm{E}(j_t(X_A)|M_{[0,t]})$, the best estimate (in a least squares sense) of an atomic observable conditioned on the measurement record.  We invite the reader to consult \cite{vanHandel:2005a,Bouten:2006a} for details on deriving the recursive form of this filter, which is governed by the (classical) stochastic differential equation (SDE)
\begin{multline} \label{eq:PiFilter}
	d\pi_t[X_A] = \pi_t[\mathcal{L}[X_A]]dt + 
		\left(\pi_t[L^{\dag}X_A + X_AL] \right.\\
		\left. - \pi_t[L^{\dag} + L]\pi_t[X_A]\right)
				\times (dM_t - \pi_t[L + L^{\dag}]dt) .
\end{multline}
We see that this is an entirely classical filter, driven by the classical measurement stream $Y_t$.  Oftentimes, it is more convenient to work with the adjoint form of the equation for $\rho_t$, which satisfies $\Tr{X_A\rho_t} = \pi_t[X_A]$ for all $X_A \in \mathcal{A}$.  The state $\rho_t$ is often called the conditional density matrix.  The SDE or stochastic master equation (SME) for $\rho_t$ is then
\begin{multline} \label{eq:RhoFilter}
	d\rho_t = -i[H,\rho_t]dt + (L\rho_tL^{\dag} - \frac{1}{2}L^{\dag}L\rho_t - \frac{1}{2}\rho_tL^{\dag}L)dt\\
			+ (L\rho_t + \rho_tL^{\dag} - \Tr{(L+L^{\dag})\rho_t}\rho_t)dW_t
\end{multline}
where the \emph{innovations process}, $dW_t = dM_t - \Tr{(L+L^{\dag})\rho}dt$, is a Wiener process that satisfies $\mathbbm{E}[dW_t] = 0 $ and It\^{o} rule $(dW_t)^2 = dt$.
\subsubsection*{Qubit Example} \label{subsec:single_qubit}
\begin{figure}[ht]
	\centering
		\includegraphics{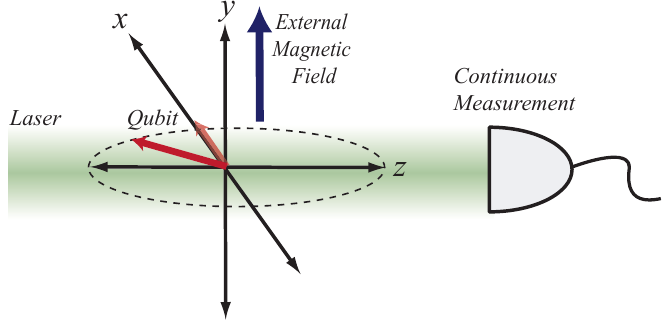}
	\caption{Continuous-measurement of single qubit precessing in an external magnetic field}
	\label{fig:schematic}
\end{figure}
\begin{figure}[t]
	\centering
		\includegraphics[scale=1]{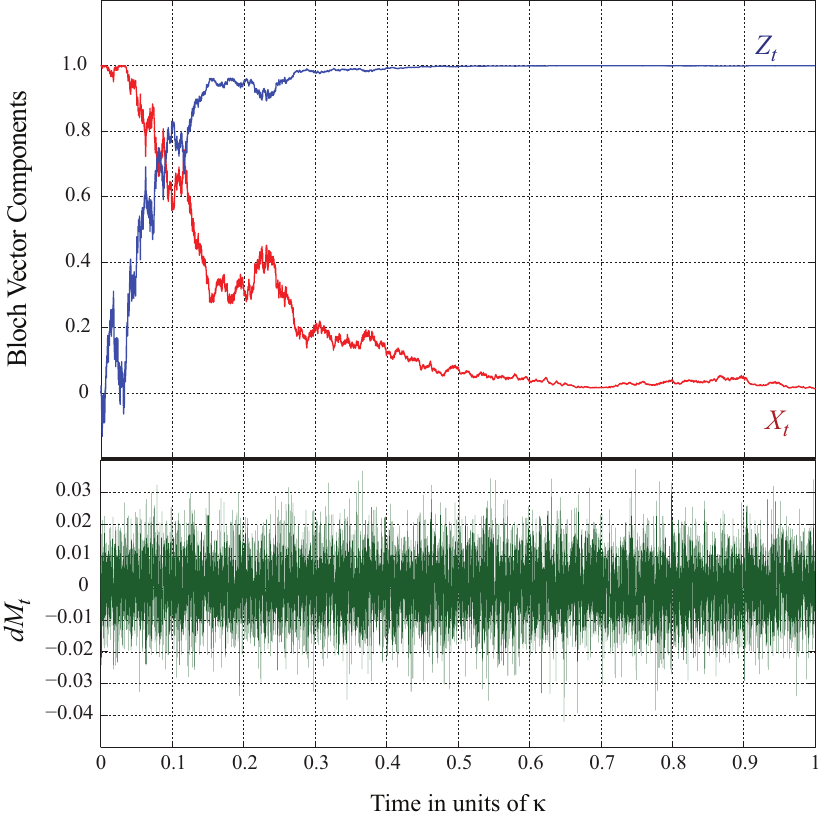}
	\caption{(Bottom) Simulated typical measurement trajectory for continuous $Z$ measurement, $\kappa = 1$, $B = 0$ (Top) Filtered values of $\pi_t[\sigma_x]$ and $\pi_t[\sigma_z]$ for simulated trajectory}
	\label{fig:typical_trajectory}
\end{figure}
Consider the setup depicted in Figure \ref{fig:schematic}.  A qubit, initially in the pure state $\ket{+x}$, precesses about a magnetic field $B$ while undergoing a continuous measurement along $z$.  In terms of the general framework, $H = B \sigma_y$ and $L = \sqrt{\kappa}\sigma_z$, where $\sqrt{\kappa}$ is the continuous measurement strength in the weak coupling limit.  We will not dwell on the underlying physical mechanism which gives rise to the $\sigma_z$ measurement, though continuous polarimetry measurements could suffice \cite{Bouten:2007a}.  Plugging into \eqref{eq:PiFilter}, the quantum filter for the Bloch vector $n_t = (\pi_t[\sigma_x],\pi_t[\sigma_y],\pi_t[\sigma_z])$ is 
\begin{align}
	d\pi_t[\sigma_x] &= 2B\pi_t[\sigma_z] dt - 2\kappa \pi_t[\sigma_x] dt
					  - 2\sqrt{\kappa}\pi_t[\sigma_x]\pi_t[\sigma_z] dW_t\\
	d\pi_t[\sigma_y] &= -2M\pi_t[\sigma_y] dt - 2\sqrt{\kappa}\pi_t[\sigma_x]\pi_t[\sigma_y] dW_t\\
	d\pi_t[\sigma_z] &= -2B\pi_t[\sigma_x]dt  + 2\sqrt{\kappa}( 1 -  \pi_t[\sigma_z]^2)dW_t
\end{align}
with innovations $dW_t = dM_t - 2\sqrt{\kappa}\pi_t[\sigma_z] dt$.  It is not difficult to verify that the quantum filter maintains pure states and that the initial state $n_0 = (1,0,0)$ remains on the Bloch circle in the $x$-$z$ plane.  Letting $\theta$ be the angle from the positive $x$-axis such that $\tan{\theta} = \pi_t[\sigma_z]/\pi_t[\sigma_x]$, we then simplify the filter to
\begin{equation} \label{eq:known_qubit_filter}
	d\theta_t = -2Bdt + \kappa \sin(2\theta_t)dt + 2\sqrt{\kappa}\cos(\theta_t)dW_t 
\end{equation}
where now $dW_t = dM_t - 2 \sqrt{\kappa}\sin\theta_t$.  Figure \ref{fig:typical_trajectory} shows a computer simulation of a typical measurement trajectory and filtered Bloch vector values when $B = 0$.  Note that in the absence of a magnetic field, the steady-states are the $\sigma_z$ eigenstates, which are reached with probabilities given by the Born rule \cite{Adler:2001a}. 
\section{Estimation of a parameter from a finite set}
\label{sec:finite}
Using the atomic system as a probe for the unknown parameter $\xi$ \footnote{This is not to be confused with the notion of the  electromagnetic field as a probe for indirect measurements of the atomic system.}, we set the atomic Hamiltonian of the quantum filter to 
\begin{equation} \label{eq:start_hamiltonian}
	H = \xi H_0 , \quad  H_0 \in \mathcal{A} .
\end{equation}
Supposing we knew the true value of the parameter, the quantum filtering equations would give us the best least-squares estimate of the atomic system conditioned on the measurements and the knowledge of dynamics induced by $\xi$ through $H$.  But given the optimality of the filter, we could equally well embed the parameter $\xi$ as a diagonal operator $\Xi$ acting on an auxiliary quantum space, after which the filter \emph{still} gives the best estimate of \emph{both} system and auxiliary space operators.  Finding the best estimate of $\xi$ conditioned on the measurements simply corresponds to integrating the equations for $\pi_t[\Xi]$.

More precisely, extend the atomic Hilbert space $\mathcal{H}_A \mapsto \mathcal{H}_{\xi}\otimes\mathcal{H}_A$ and the operator space $\mathcal{A} \mapsto \mathfrak{D}(\mathcal{H}_{\xi})\otimes\mathcal{A}$, where $\mathfrak{D}(\mathcal{H}_{\xi})$ is the set of diagonal operators on $\mathcal{H}_{\xi}$.  Assuming $\xi$ takes on $N$ possible values $\{\xi_1,\ldots, \xi_N\}$, $\dim{\mathfrak{D}(\mathcal{H}_{\xi})} = N$.   Introduce the diagonal operator 
\begin{equation}
	\mathfrak{D}(\mathcal{H}_{\xi}) \ni \Xi = \sum_{i=1}^{N} \xi_i \ketbra{\xi_i}{\xi_i}
\end{equation}
so that $\Xi\ket{\xi_i} = \xi_i\ket{\xi_i}$ with $\ket{\xi_i} \in \mathcal{H}_{\xi}$.  This allows one to generalize \eqref{eq:start_hamiltonian} as 
\begin{equation}
	H \mapsto \Xi\otimes H_0 \in \mathfrak{D}(\mathcal{H}_{\xi})\otimes\mathcal{A} .
\end{equation}
Any remaining atomic operators $X_A \in \mathcal{A}$ act as the identity on the auxiliary space, i.e. $I\otimes X_A$.  
Given these definitions, the derivation of the quantum filtering equation remains essentially unchanged, so that the filter in either the operator form of \eqref{eq:PiFilter} or the adjoint form of \eqref{eq:RhoFilter} is simply updated with the extended forms of operators given in the last paragraph.  

Since $\xi$ is a classical parameter, we require that the reduced conditional density matrix $(\rho_{\xi})_t = \operatorname{Tr}_{\mathcal{H}_A}{(\rho_t)}$ be diagonal in the basis of $\Xi$.  Thus we can write
\begin{equation} \label{eq:reduced_state}
	(\rho_{\xi})_t = \sum_{i=1}^N p_t^{(i)} \ketbra{\xi_i}{\xi_i}
\end{equation}
where
\begin{multline}
	p_t^{(i)} \equiv \Tr{(\ketbra{\xi_i}{\xi_i} \otimes I)\rho_t} 
		\equiv \pi_t[\ketbra{\xi_i}{\xi_i}\otimes I]\\
		\equiv \mathbbm{E}[\ketbra{\xi_i}{\xi_i} \otimes I | M_{[0,t]}] 
		\equiv P(\xi = \xi_i | M_{[0,t]}) .
\end{multline}
Then $p_t^{(i)}$ is precisely the conditional probability for $\xi$ to have the value $\xi_i$ and the set $\{p_t^{(i)}\}$ gives the discrete conditional distribution of the random variable represented by $\Xi$.  Similarly, by requiring operators to be diagonal in $\mathcal{H}_{\xi}$, we ensure that they correspond to classical random variables.  In short, we have simply embedded filtering of a truly classical random variable in the quantum formalism.

The fact that both states and operators are diagonal in the auxiliary space suggests using an ensemble form for filtering.  As such, consider an ensemble consisting of a weighted set of $N$ conditional atomic states,  each state evolved under a different $\xi_i$.  Later, in section \ref{sec:infinite}, we will call each ensemble member a \emph{quantum particle}. For now, we explicitly write the conditional quantum state as
\begin{equation} \label{eq:finitedim:ensembleform}
	\rho_t^{E} = \sum_{i = 1}^{N} p^{(i)}_t \ketbra{\xi_i}{\xi_i} \otimes \rho^{(i)}_t 
\end{equation}
where $\rho^{(i)}_t$ is a density matrix on $\mathcal{H}_A$.  The reduced state, $\operatorname{Tr}_{\mathcal{H}_A}{(\rho_t^{E})}$, is clearly diagonal in the basis of $\Xi$.  Using the extended version of the adjoint quantum filter in \eqref{eq:RhoFilter}, one can derive the \emph{ensemble quantum filtering equations}
\begin{subequations} \label{eq:finitedim:ensemblefilter}
	\begin{align}
			d\rho_t^{(i)} &= -i[\xi_i H_0, \rho_t^{(i)}]dt   
						   + (L\rho_t^{(i)}L^{\dag} - \frac{1}{2}L^{\dag}L\rho_t^{(i)} - \frac{1}{2}\rho_t^{(i)}L^{\dag}L)dt
							\nonumber \\
					      & + \left(L\rho_t^{(i)} + \rho_t^{(i)}L^{\dag}
					 		- \Tr{(L+L^{\dag})\rho_t^{(i)}}\rho_t^{(i)}\right)dW_t \label{eq:finitedim:ensemblefilter:rho}\\
			dp_t^{(i)} &= \left(\Tr{(L+L^{\dag})\rho_t^{(i)}} - \Tr{I\otimes(L+ L^{\dag})\rho_t^{E}}\right)p_t^{(i)}dW_t
							\label{eq:finitedim:ensemblefilter:prob}\\
			dW_t &= dM_t - \Tr{I\otimes ( L +  L^{\dag})\rho_t^{E}}dt \label{eq:finitedim:ensemblefilter:innov}
	\end{align}
\end{subequations}
We see that each $\rho^{(i)}_t$ in the ensemble evolves under a quantum filter with $H = \xi_i H_0$ and is coupled to other ensemble members through the innovation factor $dW_t$, which depends on the ensemble expectation of the measurement observable.  Note that one can incorporate any prior knowledge of $\xi$ in the weights of the initial distribution $\{p_0^{(i)}\}$.

The reader should not be surprised that a similar approach would work for estimating more than one parameter at a time, such as three cartesian components of an applied magnetic field.  One would introduce an auxiliary space for each parameter and extend the operators in the obvious way.  The ensemble filter would then be for a joint distribution over the multi-dimensional parameter space.  Similarly, one could use this formalism to distinguish initial states, rather than parameters which couple via the Hamiltonian.  For example, in the case of state discrimination, one would introduce an auxiliary space which labels the possible input states, but does not play any role in the dynamics.  The filtered weights would then be the probabilities to have been given a particular initial state. In fact, using a slightly different derivation, Jacobs derived equations similar to \eqref{eq:finitedim:ensemblefilter} for the case of binary state discrimination \cite{Jacobs:2006a}.  Yanagisawa recently studied the general problem of retrodiction or ``smoothing'' of quantum states \cite{Yanagisawa:2007a}.  In light of his work and results in the following section, the retrodictive capabilities of quantum filtering are very limited without significant prior knowledge or feedback.   
\subsection{Conditions for convergence}
Although introducing the auxiliary parameter space does not change the derivation of the quantum filter, it is not clear how the initial uncertainty in the parameter will impact the filter's ability to ultimately track to the correct value.  Indeed, outside of anecdotal numerical evidence (which we will presently add to), there has been little formal consideration of the sensitivity of the quantum filter to the initial state estimate.  Recently, van Handel presented a set of conditions which determine whether the quantum filter will asymptotically track to the correct state independently of the assumed initial state \cite{vanHandel:2008a}.  Since we have embedded parameter estimation in the state estimation framework, such stability then determines whether the quantum filter can asymptotically track to the true parameter, i.e. whether $\lim_{t \rightarrow \infty} p_{t}^{(j)} =  \delta_{ij}$ when $\xi = \xi_i$.  In this section, we present van Handel's results in the context of our parameter estimation formalism and present a simple check of asymptotic convergence of the parameter estimate.  We begin by reviewing the notions of absolute continuity and observability.

In the general stability problem, let $\rho_1$ be the true underlying state and $\rho_2$ be the initial filter estimate.  We say that $\rho_1$ is \emph{absolutely continuous} with respect to $\rho_2$, written $\rho_1 \ll \rho_2$, if and only if $\operatorname{ker} \rho_1 \supset \operatorname{ker} \rho_2$.  In the context of parameter estimation, we assume that we know the initial atomic state exactly, so that $\rho_1 \ll \rho_2$ as long as the reduced states satisfy $\rho_1^{E} \ll \rho_2^{E}$.  Since these reduced states are simply discrete probability distributions, $\{(p_t^{i})_1\}$ and $\{(p_t^{i})_2\}$, this is just the standard definition of absolute continuity in classical probability theory.  In our case, the true state has ${(p_{t=0}^{(j)})}_{1} = \delta_{ij}$ if the parameter has value $\xi_i$.  Thus, as long as our estimate has non-zero weight on the $i$-th component, $\rho_1 \ll \rho_2$.  This is trivially satisfied if ${(p_{t=0}^{(j)})}_{2} \neq 0$ for all $j$.

The other condition for asymptotic convergence is that of observability.  A system is \emph{observable} if one can determine the exact initial atomic state given the entire measurement record over the infinite time interval.  Observability is then akin to the ability to distinguish any pair of initial states on the basis of the measurement statistics alone.  Recall the definition of the Lindblad generator in \eqref{eq:lindblad_generator} and further define the operator $\mathcal{K}[X_A] = L^{\dag}X_a + X_aL$.  Then according to Proposition 5.7 in \cite{vanHandel:2008a}, the observable space $\mathcal{O}$ is defined as the smallest linear subspace of $\mathcal{A}$ containing the identity and which is invariant under the action of $\mathcal{L}$ and $\mathcal{K}$.  The filter is observable if and only if $\mathcal{A} = \mathcal{O}$, or equivalently $\dim{\mathcal{A}} = \dim{\mathcal{O}}$.  

In the finite-dimensional case, van Handel presents an iterative procedure for constructing the observable space.  Define the linear spaces $\mathcal{Z}_n \subset \mathcal{A}$ as
\begin{equation}
	\begin{split}
		\mathcal{Z}_0 &= \operatorname{span}\{I\}\\
		 \mathcal{Z}_n &= \operatorname{span}\{\mathcal{Z}_{n-1},
								\mathcal{L}[\mathcal{Z}_{n-1}], \mathcal{K}[\mathcal{Z}_{n-1}]\},\quad n > 0
	\end{split}
\end{equation}
The procedure terminates when $\mathcal{Z}_n = \mathcal{Z}_{n+1}$, which is guaranteed for some finite $n = m$, as the dimension of $\mathcal{Z}_n$ cannot exceed the dimension of the ambient space $\mathcal{A}$.  Moreover, the terminal $\mathcal{Z}_m = \mathcal{O}$, so that using a Gram-Schmidt procedure, one can iteratively find a basis for $\mathcal{O}$ and easily compute its dimension.  Note that for operators $A$ and $B$, the inner-product $\langle A, B \rangle$ is the Hilbert-Schmidt inner product $\Tr{A^{\dag} B}$.

Given these definitions, one has the following theorem for filter convergence and corollary for parameter estimation.

\begin{thm} (Theorem 2.5 in \cite{vanHandel:2008a}) Let $\pi_t^{\rho_i}(X_A)$ be the evolved filter estimate, initialized under state $\rho_i$.  If the system is observable and $\rho_1 \ll \rho_2$, the quantum filter is asymptotically stable in the sense that
\begin{equation}
	\abs{\pi_t^{\rho_1}(X_A) - \pi_t^{\rho_2}(X_A)}_{M^{\rho_1}_{[0,t]}} \stackrel{t \rightarrow\infty}
						{\longrightarrow} 0 \quad \forall X_a \in \mathcal{A}
\end{equation} 
where the convergence is under the observations generated by $\rho_1$.
\end{thm}

One could use this theorem to directly check the stability of the quantum filter for parameter estimation, using the extended forms of operators in $\mathcal{L}$ and $\mathcal{K}$ and being careful that the observability condition is now $\dim{\mathcal{O}} = \dim{\mathcal{D}{(\mathcal{H}_\xi)}\otimes\mathcal{A}}$.  However, the following corollary relates the observability of the parameter filter to the observability of the related filter for a known parameter.  Combined with the discussion of extending the absolute continuity condition, this then gives a simple check for the stability of the parameter filter.  

\begin{cor}
Consider a parameter $\xi$ which takes on one of $N$ distinct positive real values $\{\xi_i\}$.  If the quantum filter with known parameter is observable, then the corresponding extended filter for estimation of $\xi$ is observable.
\end{cor}
\begin{proof}
In order to satisfy the observability condition, we require $\dim{\mathcal{O}} = Nr$, where we have set $\dim{\mathcal{A}} = r$ and used the fact that $\dim{\mathfrak{D}(\mathcal{H}_\xi)} = N$.  Given that the filter for a known parameter is observable, its observable space coincides with $\mathcal{A}$ and has an orthogonal operator basis $\{A_i\}$, where we take $A_0 = I$.

Similarly, consider the $N$-dimensional operator space $\mathfrak{D}(\mathcal{H}_\xi)$.  If $\{\xi_i\}$ are distinct, any set of the form
\begin{equation}
	\{\Xi^{k_1}, \Xi^{k_2}, \ldots, \Xi^{k_N}\}, k_i \in \mathbbm{N}, k_i \neq k_j \text{ if } i \neq j
\end{equation}
is linearly independent, since the corresponding generalized Vandermonde matrix
\begin{equation}
	V_\xi = \begin{pmatrix}
		\xi_1^{k_1} & \xi_1^{k_2} & \ldots &  \xi_1^{k_N} \\
		\vdots & \vdots & \ddots & \vdots  \\
		\xi_N^{k_1} & \xi_N^{k_2} &\ldots & \xi_N^{k_N}
	\end{pmatrix}
\end{equation}
has linearly independent columns \cite{Gantmakher:2000a}.  

Following the iterative procedure, we construct the observable space for the parameter estimation filter starting with $I \otimes A_0$, which is the identity in the extended space. We then iteratively apply $\mathcal{L}$ and $\mathcal{K}$ until we have an invariant linear span of operators.  The only non-trivial operator on the auxiliary space comes from the Hamiltonian part of the Lindblad generator, which introduces higher and higher powers of the diagonal matrix $\Xi$.  Since $\dim{\mathfrak{D}{(\mathcal{H}_\xi)}\otimes\mathcal{A}}$ is finite, this procedure must terminate. The resulting observable space can be decomposed into subspaces
\begin{equation}
\mathcal{O}_{i} = \{\Xi^{k_i^j} \otimes A_i\}, \quad i = 1,\ldots, r \quad k_i^j \in \mathbbm{N}  
\end{equation}
where $k_i^j$ is some increasing sequence of non-negative integers which correspond to the powers of $\Xi$ that are introduced via the Hamiltonian.  Note that the specific values of $k_i^j$ depend on the commutator algebra of $H_0$ and the atomic-space operator basis $\{A_i\}$.  Regardless, since the Hamiltonian in $\mathcal{L}$ can always add more powers of $\Xi$, the procedure will not terminate until $\mathcal{O}_i$ is composed of a largest linearly independent set of powers of $\Xi$.  This set has at most $N$ distinct powers of $\Xi$, since it cannot exceed the dimension of the auxiliary space.  Given that any collection of $N$ powers of $\Xi$ is linearly independent, this means once we reach a set of $N$ powers $k_i^j$, the procedure terminates and $\dim{\mathcal{O}_i} = N$.   Since $\mathcal{O}$ has $r$ subspaces $\mathcal{O}_i$, each of dimension $N$, $\dim{\mathcal{O}} =  Nr$ as desired and the observability condition is satisfied.
\end{proof}

Although these conditions provide a simple check, we would like to stress that they do not determine how quickly the convergence occurs, which will depend on the specifics of the problem at hand.  Additionally, as posed, the question of observability is a binary one.  One might expect that some unobservable systems are nonetheless ``more observable'' than others or simply that unobservable systems might still be useful for parameter estimation.  Given the corollary above, one can see that this may occur if a single parameter $\xi_j = 0$.  Then $V_{\xi}$ has a row of all zeros, so that the maximal dimension of a set of linearly independent powers of $\Xi$ is $N - 1$.  Similarly, if one allows both positive and negative real-valued parameters, the properties of $V_{\xi}$ are not as obvious, though in many circumstances, having both $\xi_i$ and $-\xi_i$ renders the system unobservable.  We explore these nuances in numerical simulations presented in the following section.  
\subsubsection*{Qubit Example}
\begin{figure*}[bt]
	\centering
		\includegraphics[scale=1]{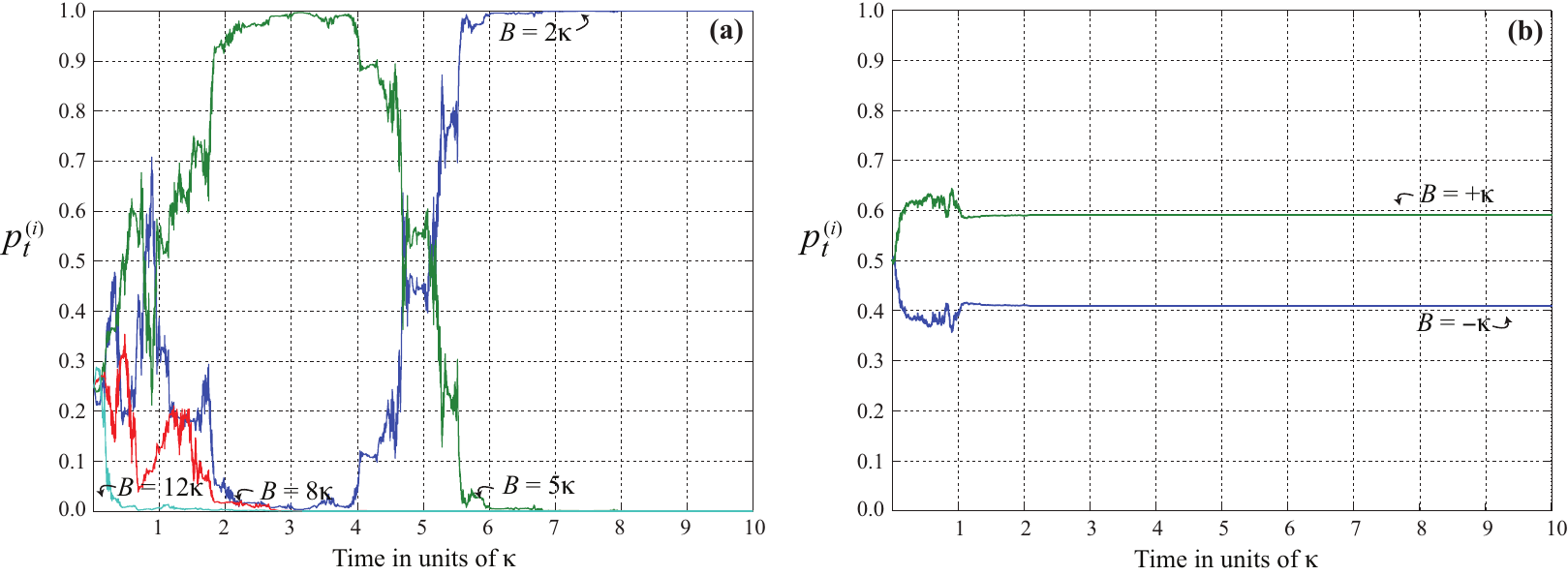}
	\caption{(a) Filtered $p_t^{(i)}$ for $B \in \{ 2\kappa, 5\kappa, 8\kappa, 12\kappa\}$.  The filter tracks to the true underlying value of $B = 2\kappa$ (b) Filtered $p_t^{(i)}$ for $B \in \{ -\kappa, +\kappa\}$.  The filter does not track to $B = +\kappa$ with probability one, though it is the most probable parameter value.}
	\label{fig:discrete:combinedruns}
\end{figure*}
Consider using the single qubit magnetometer of Section \ref{subsec:single_qubit} as a probe for the magnetic field $B$.  Since the initial state is restricted to the $x$-$z$ plane, the $y$ component of the Bloch vector is always zero and thus is not a relevant part of the atomic observable space, which is spanned by $\{I,\sigma_x,\sigma_z\}$.  In some sense, the filter with known $B$ is trivially observable, since we assume the initial state is known precisely.  

When $B$ is unknown, the ensemble parameter filter is given by
\begin{subequations} \label{eq:qubit_ensemble}
	\begin{align} 
			d\theta_t^{(i)} &= -2B_idt + \kappa \cos(\theta_t^{(i)})
									(\sin(\theta_t^{(i)}) - 2\expect{\sigma_z}^{E})dt \nonumber\\
							&	+ 2\sqrt{\kappa}\cos(\theta_t^{(i)})dW_t \\
			dp^{(i)}_t &= 2\sqrt{\kappa}(\sin(\theta_t^{(i)}) - \expect{\sigma_z}^{E}) p^{(i)}_t dW_t
	\end{align}
\end{subequations}
where $dW_t = dM_t -2 \sqrt{\kappa} \expect{\sigma_z}^{(E)}$ and $\expect{\sigma_z}^{E} = \sum_{i} p_t^{(i)}\sin(\theta_t^{(i)})$.  We simulated this filter by numerically integrating the quantum filter in \eqref{eq:known_qubit_filter} using a value for $B$ uniformly chosen from the given ensemble of potential $B$ values.  This generates a measurement current $dM_t$, which is then fed into the ensemble filter of \eqref{eq:qubit_ensemble}.  For all simulations, we set $\kappa = 1$ and used a simple It\^{o}-Euler integrator with a step-size $dt = 10^{-5}$ \cite{Kloeden:1992a}.

Figure \ref{fig:discrete:combinedruns}(a) shows a simulation of a filter for the case $B \in \{ 2\kappa, 5\kappa, 8\kappa, 12\kappa\}$.  The filter was initialized with a uniform distribution, $p^{(i)}_0 = 1/4$.  For the particular trajectory shown, the true value of $B$ was $2\kappa$ and we see that the filter successfully tracks to the correct $B$ value.  This is not surprising, given that the potential values of $B$ are positive and distinct, thus satisfying the convergence corollary.  It is also interesting to note that the filter quickly discounts the probabilities for $8\kappa,12\kappa$, which are far from the true value.  Conversely, the filter initially favors the incorrect $B = 5\kappa$ value before honing in on the correct parameter value.

In Figure \ref{fig:discrete:combinedruns}(b), we show a simulation for the case of $B \in \{+\kappa, -\kappa\}$, which does not satisfy the convergence corollary.  In fact, using the iterative procedure, one finds the observable space is spanned by $\{I\otimes I, I \otimes \sigma_z, B \otimes \sigma_x, B^2 \otimes I, B^2 \otimes \sigma_z, B^3 \otimes \sigma_x\}$.  But since $B = \left(\begin{smallmatrix} \kappa & 0\\ 0 & -\kappa \end{smallmatrix}\right)$, $B^2 = \kappa^2 I$ so that only 3 of the 6 operators are linearly independent.  Although the filter does not converge to the true underlying value of $B = +\kappa$, it does reach a steady-state that weights the true value of $B$ more heavily.  Simulating 100 different trajectories for the filter, we observed 81 trials for which the final probabilities were weighted more heavily towards the true value of $B$.  This confirms our intuition that the binary question of observability does not entirely characterize the performance of the parameter filter.  

Figure \ref{fig:discrete:convergence} shows the rate of convergence of filters meant to distinguish different sets of $B$.  The rate of convergence is defined as the ensemble average of the random variable
\begin{equation}
	I_{\alpha} = \begin{cases}
		1, &\text{ if $p_t^{(i)} > \alpha $ for any $i$}\\
		0, &\text{ otherwise }
	\end{cases} .
\end{equation}
Although any individual run might fluctuate before converging to the underlying $B$ value, the average of $I_\alpha$ over many runs should give some sense of the rate at which these fluctuations die down.  For the simulation shown, we set $\alpha = 0.95$ and averaged $I_{0.95}$ over 1000 runs for two different cases---either all possible $B$ values are greater than $\kappa$ or all are less than $\kappa$.  As shown in the plot, the former case shows faster convergence since the $B$ field drives the dynamics more strongly than the measurement process, which in turn makes the trajectories of different ensemble members more distinct.  Of course, one cannot make the measurement strength too weak since we need to learn about the system evolution.
\begin{figure}[hb]
	\centering
		\includegraphics[scale=1]{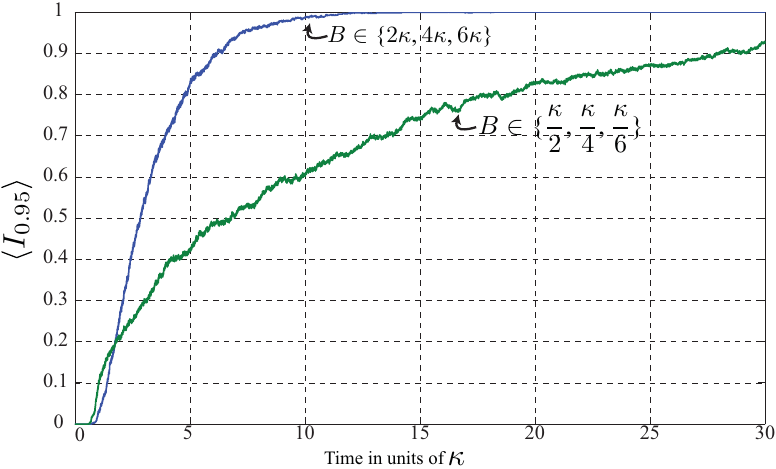}
	\caption{Rate of convergence ($I_{0.95}$), averaged over 1000 trajectories.  The filters are for cases when possible $B$ values are either all larger or all smaller than the measurement strength $\kappa$.  }
	\label{fig:discrete:convergence}
\end{figure}
\section{Quantum Particle Filter}
\label{sec:infinite}
Abstractly, developing a parameter estimator in the continuous case is not very different than in the finite dimensional case.  One can still introduce an auxiliary space $\mathcal{H}_{\xi}$, which is now infinite dimensional.  In this space, we embed the operator version of $\xi$ as 
\begin{equation}
	\mathfrak{D}(\mathcal{H}_{\xi}) \ni \Xi = \int d\xi \xi \ketbra{\xi}{\xi} ,
\end{equation}
where $\Xi\ket{\xi} = \xi\ket{\xi}$ and $\braket{\xi}{\xi'} = \delta(\xi - \xi')$.  Again, by extending operators appropriately, the filters in \eqref{eq:PiFilter} and \eqref{eq:RhoFilter} become optimal parameter estimation filters.  We generalize the conditional ensemble state of \eqref{eq:finitedim:ensembleform} to
\begin{equation} \label{eq:infinite:rho}
		\rho_t^{E} = \int d\xi p_t(\xi) \ketbra{\xi}{\xi} \otimes \rho^{(\xi)}_t ,
\end{equation}
where $p_t(\xi) \equiv  P(\xi | M_{[0,t]})$ is the continuous conditional probability density.  Although the quantum filter provides an exact formula for the evolution of this density, calculating it is impractical, as one cannot exactly represent the continuous distribution on a computer.  The obvious approximation is to discretize the space of parameter values and then use the ensemble filter determined by \eqref{eq:finitedim:ensemblefilter}; indeed such an approach is very common in classical filtering theory and encompasses a broad set of Monte Carlo methods called particle filters \cite{Doucet:2001,Arulampalam:2002a}.

The inspiration for particle filtering comes from noting that any distribution can be approximated by a weighted set of point masses or \emph{particles}.  In the quantum case, we introduce a \emph{quantum particle} approximation of the conditional density in \eqref{eq:infinite:rho} as
\begin{equation}
	p_t(\xi) \approx \sum_{i = 1}^{N} p^{(i)}_t \delta(\xi - \xi_i) .
\end{equation}
The approximation can be made arbitrarily accurate in the limit of $N \rightarrow \infty$.  Plugging this into \eqref{eq:infinite:rho}, we recover precisely the form for the discrete conditional state given in \eqref{eq:finitedim:ensembleform}.  Accordingly, the quantum particle filtering equations are identical to those of the ensemble filter given in \eqref{eq:finitedim:ensemblefilter}.  The only distinction here is in the initial approximation of the space of parameter values.  Thus the basic quantum particle filter simply involves discretizing the parameter space,  then integrating the filter according to the ensemble filtering equations.  

The basic particle filter suffers from a degeneracy problem, in that all but a few particles may end up with negligible weights $p^{(i)}_t$.  This problem is even more relevant when performing parameter estimation, since the set of possible values for $\xi$ are fixed at the outset by the choice of discretization.  Even if a region in parameter space has low weights, its particles take up computational resources, but contribute little to the estimate of $\xi$.  More importantly, the ultimate precision of the parameter estimate is inherently limited by the initial discretization; we can never have a particle whose parameter value $\xi_i$ is any closer to the true value $\xi$ than the closest initial discretized value.

In order to circumvent these issues, we adopt the kernel resampling techniques of Liu and West \cite{Liu:2001a}.  The idea is to replace low weight particles with new ones concentrated in high weight regions of parameter space.  One first samples a source particle from the discrete distribution given by the weights $\{p^{(i)}_t\}$, ensuring new particles come from more probable regions of parameter space.  Given a source particle, we then create a child particle by sampling from a Gaussian kernel centered near the source particle.  By repeating this procedure $N$ times, we create a new set of particles which populate more probable regions of parameter space.  Over time, this adaptive procedure allows the filter to move away from unimportant regions of parameter space and more finely explore the most probable parameter values.

The details of the adaptive filter lie in parameterizing and sampling from the Gaussian kernel.  Essentially, we are given a source particle, characterized by $\ketbra{\xi_i}{\xi_i}$ and $\rho_t^{(i)}$, and using the kernel, create a child particle, characterized by $\ketbra{\tilde{\xi}_i}{\tilde{\xi}_i}$ and $\tilde{\rho}_t^{(i)}$.  One could attempt to sample from a multi-dimensional Gaussian over both the parameter and atomic state components, but ensuring that the sampled $\tilde{\rho}_t^{(i)}$ is a valid atomic state would be non-trivial in general.  There will be some cases, including the qubit example in the following section, where the atomic state is conveniently parameterized for Gaussian resampling.  But for clarity in presenting the general filter, we will create a child particle with the same atomic state as the parent particle.  

Under this assumption, the Gaussian kernel for parent particle $i$ is characterized by a mean $\mu^{(i)}$ and variance ${\sigma^2}^{(i)}$, both defined over the one dimensional parameter space.  Rather than setting the mean of this kernel to the parameter value of the parent, Liu and West suggest setting
\begin{equation} \label{eq:kernel:mean}
	\mu^{(i)} = a \xi_i + (1 - a) \bar{\xi} , \quad a \in [0,1]
\end{equation}
where $\bar{\xi} = \sum_i p_t^{(i)}\xi_i$ is the ensemble mean.  The parameter $a$ is generally taken to be close to one and serves as a mean reverting factor.  This is important because simply resampling from Gaussians centered at $\xi_i$ results in an overly dispersed ensemble relative to the parent ensemble.  The kernel variance is set to
\begin{equation} \label{eq:kernel:var}
	{\sigma^2}^{(i)} = h^2 V_t , \quad h \in [0,1]
\end{equation}
where $V_t = \sum_i p_t^{(i)}( \xi_i - \bar{\xi})^2$ is the ensemble variance and $h$ is the smoothing parameter.  It is generally a small number chosen to scale with $N$, so as to control how much kernel sampling explores parameter space.  While $a$ and $h$ can be chosen independently,  Liu and West relate them by $h^2 = 1 - a^2$, so that the new sample does not have an increased variance.
  
Of course, it would be computationally inefficient to perform this resampling strategy at every timestep, especially since there will be many steps where most particles have non-negligible contributions to the parameter estimate.  Instead, we should only resample if some undesired level of degeneracy is reached.  As discussed by Arulampalam et al. \cite{Arulampalam:2002a}, one measure of degeneracy is the effective sample size
\begin{equation}
	N_\text{eff} = \frac{1}{\sum_{i=1}^N (p_t^{(i)})^2} .
\end{equation}
At each timestep, we then resample if the ratio $N_\text{eff}/N$ is below some given threshold.  We are not aware of an optimal threshold to chose in general, but the literature suggest $2/3$ as a rule of thumb \cite{Doucet:2001}.

Altogether, the \emph{resampling quantum particle filter} algorithm proceeds as follows:
\begin{description}
	\item[Initialization] for $i = 1,\ldots, N$:
	 \begin{enumerate}
	 	\item Sample $\xi_i$ from the prior parameter distribution.
		\item Create a quantum particle with weight $p_t^{(i)} = 1/N$, parameter state $\ketbra{\xi_i}{\xi_i}$ and atomic state $\rho_0^{(i)} = \rho_0$, where $\rho_0$ is the known initial atomic state.
	 \end{enumerate}
	\item[Repeat] for all time:
		\begin{enumerate}
			\item Update the particle ensemble by integrating a timestep of the filter given in \eqref{eq:finitedim:ensemblefilter}.
			 \item If $N_\text{eff}/N$ is less than the target threshold, create a new particle ensemble:
			  \begin{description}
			  	\item[Resample] for $i = 1,\ldots, N$:
				 \begin{enumerate}
				 	\item Sample an index $i$ from the discrete density $\{p_t^{(i)}\}$.
					\item Sample a new parameter value $\tilde{\xi}_i$ from the Gaussian kernel with mean $\mu^{(i)}$ and variance ${\sigma^2}^{(i)}$ given by \eqref{eq:kernel:mean} and \eqref{eq:kernel:var}. 
					\item Add a quantum particle to the new ensemble with weight $p_t^{(i)} = 1/N$, parameter state $\ketbra{\tilde{\xi}_i}{\tilde{\xi}_i}$ and atomic state $\rho_t^{(i)} = \rho_t^{(i)}$
				 \end{enumerate}
			  \end{description}
		\end{enumerate}
\end{description}

Unfortunately, checking asymptotic convergence of the filter is more involved in the continuous-valued case, as the observability and absolute continuity conditions require extra care in infinite dimensions.  However, given that the quantum particle filter actually works on a discretized space, in practice we can simply use the results we had for the finite-dimensional case.  As before, we note that one can generalize the quantum particle filter to multidimensional parameters by using a multi-dimensional Gaussian kernel.  One might also consider using alternate kernel forms, such as a regular grid which has increasingly finer resolution with each resampling stage.  We will not consider such extensions here.  
\subsubsection*{Qubit Example}
We now consider a resampling quantum particle filter for the qubit magnetometer introduced earlier in the paper.  As hinted at in the previous section, since the qubit state is parameterized by the continuous variable $\theta_t$, we can easily resample both the magnetic field $B_i$ and state $\theta^{(i)}$ using a two-dimensional Gaussian kernel for $(\tilde{B}_i,\tilde{\theta}^{(i)})$, with mean vector and covariance matrix given by generalizations of \eqref{eq:kernel:mean} and \eqref{eq:kernel:var}.  Since different values of $B$ result in different state evolutions, resampling both the state and magnetic field values should result in child particles that are closer to the true evolved state. 

\begin{figure}[t]
	\centering
		\includegraphics[scale=0.5]{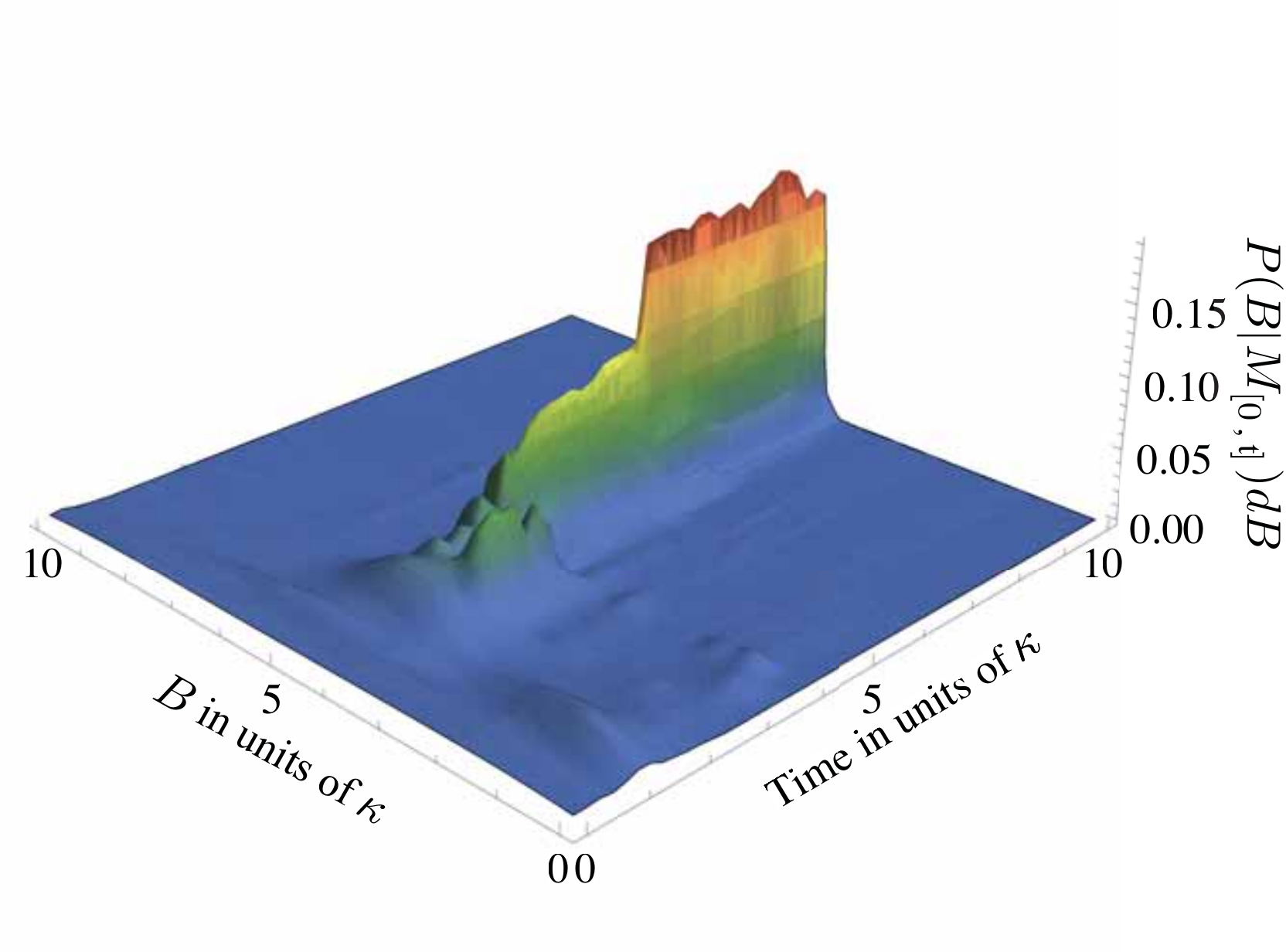}
	\caption{Kernel density reconstruction of $p_t(B)dB =  P(B | M_{[0,t]})dB$ for $N = 1000$ particle filter set with $dB = 10\kappa/150$, $a = 0.98$, $h = 10^{-3}$ and resampling threshold of $2/3$.  The true magnetic field was $B = 5\kappa$.}
	\label{fig:densityPlot}
\end{figure}

Figure \ref{fig:densityPlot} shows a typical run of the quantum particle filter for $N = 1000$ particles.  The true $B$ value was $5\kappa$ and the prior distribution over $B$ was taken to be uniform over the interval $[0,10\kappa]$.  As before, we used an It\^{o}-Euler integrator with a step-size of $dt = 10^{-5}\kappa$.  Note that both the timespan of integration and the potential values of $B$ range from $0$ to $10\kappa$ in our units.  The resampling parameters were $a = 0.98$, $h = 10^{-3}$ and resampling threshold $2/3$.  Note that we chose not to use Liu and West's relation between $a$ and $h$.   

In order to generate the figure, each particle's weight and parameter values were stored at 50 equally spaced times over the integration timespan.  Using Matlab's \texttt{ksdensity} function, these samples were then used to reconstruct $p_t(B)$ via a Gaussian kernel density estimate of the distribution.  The resulting kernel density estimate was then evaluated at 150 equally spaced $B$ values in the range $[0,10\kappa]$, which we plotted as $p_t(B)dB$ with $dB = 10\kappa/150$.  As is seen in the figure, after some initial multi-modal distributions over parameter space, the filter hones in on the true value of $B = 5\kappa$.  For the simulation shown, the final estimate was $\hat{B} = 5.03\kappa$ with uncertainty $\sigma_{\hat{B}} = 0.18\kappa$.  The filter resampled 7 times over the course of integration.    
\section{Conclusion}
\label{sec:conclude}
We have presented practical methods for single-shot parameter estimation via continuous quantum measurement.  By embedding the parameter estimation problem in the standard quantum filtering problem, the optimal parameter filter is given by an extended form of the standard quantum filtering equation.  For parameters taking values in a finite set, we gave conditions for determining whether the parameter filter will asymptotically converge to the correct value.  For parameters taking values from an infinite set, we introduced the quantum particle filter as a computational tool for suboptimal estimation.  Throughout, we presented numerical simulations of our methods using a single qubit magnetometer.  

Our techniques should generalize straightforwardly for estimating time-dependent parameters and to a lesser extent, estimating initial state parameters.  The binary state discrimination problem studied by Jacobs \cite{Jacobs:2006a} is one such example and his approach is essentially a special case of our ensemble parameter filter.  We caution that the utility of initial state or parameter estimation depends heavily on the observability and absolute continuity of the problem at hand.  Future extensions of our work include exploring alternate resampling techniques for the quantum particle filter and developing feedback strategies for improving the parameter estimate.  More broadly, we believe there is much to be learned from classical control and parameter estimation theories. 

\begin{acknowledgements}
We thank Rob Cook for many valuable discussions. This work was supported by the NSF (PHY-0639994) and the AFOSR (FA9550-06-01-0178).
\end{acknowledgements}


\end{document}